\documentclass[a4paper,11pt]{article}

\usepackage{amsmath,amsthm}
\usepackage{amssymb,enumerate,url}
\usepackage{tikz}
\usepackage{authblk}

\newtheorem{lemma}{Lemma}
\newtheorem{corollary}{Corollary}
\newtheorem{theorem}{Theorem}
\newtheorem{example}{Example}

\newcommand{\D}{\mathcal{D}}
\newcommand{\G}{\mathcal{G}}
\newcommand{\N}{\mathbb{N}}
\newcommand{\T}{\mathcal{T}}
\newcommand{\cT}{\mathcal{C}}
\newcommand{\type}{\mathsf{type}}
\newcommand{\val}{\mathsf{val}}

\begin{document}

\title{Universal Tree Source Coding Using Grammar-Based Compression\footnote{This work has been supported by the DFG research project
LO 748/10-1 (QUANT-KOMP).}}

\author{Danny Hucke and Markus Lohrey}
\affil{Universit\"at Siegen, Germany}

\maketitle

\begin{abstract}
We apply so-called tree straight-line programs to the problem of universal source coding for binary trees.
We derive an upper bound on the  maximal pointwise redundancy (or worst-case redundancy) that improve
previous bounds on the average case redundancy obtained by Zhang, Yang, and Kieffer using directed acyclic graphs. Using this,
we obtain universal codes for new classes of  tree sources.
\end{abstract}

\section{Introduction}

Universal source coding for finite sequences over a finite alphabet $\Sigma$ (i.e., strings over $\Sigma$) is a well-established topic of information theory.
Its goal is to find prefix-free lossless codes that are universal (or optimal) for classes of information sources.
In a series of papers, Cosman, Kieffer, Nelson, and Yang developed grammar-based codes that are universal
for the class of finite state sources \cite{KiYa00,KiYa02,KiefferYNC00,YangK00}. Grammar-based compression works in two steps: In a first
step, from a given input string $w \in \Sigma^*$ a context-free grammar $\G_w$ that produces only the string $w$
is computed. Context-free grammars that produce exactly one string are also known as {\em straight-line 
programs}, briefly SLPs, and are currently an active topic in text compression and algorithmics on compressed texts,
see \cite{Loh12survey} for a survey. In a second step, the SLP $\G_w$ is encoded by a binary string $B(\G_w)$.  There exist several
algorithms that compute from a given input string $w$ of length $n$ an SLP $\G_w$ of size $O(n / \log n)$ (the size
of an SLP is the total number of symbols in all right-hand sides of the grammar) \cite{KiYa00}; the best known example
is probably the LZ78 algorithm \cite{ZiLe78}. 
By combining any of these algorithms with the binary encoder $B$ for SLPs from \cite{KiYa00}, one obtains a
grammar-based encoder  $E : \Sigma^* \to \{0,1\}^*$,
whose worst case redundancy for input strings of length $n$ is
bounded by $O(\log \log n / \log n)$ for every finite state information source over the alphabet $\Sigma$. 
Here, the worst case redundancy for strings of length $n$ 
is defined as
$$
\max_{w \in \Sigma^n, P(w)>0} n^{-1} \cdot (|E(w)| + \log_2 P(w)),
$$
where $P(w)$ is the probability that the finite state information source emits $w$.
Thus, the worst case redundancy measures the maximal additive deviation of the code length from the self information, normalized by the length of 
the source string.

Over the last few years, we have seen increasing efforts aiming to extend universal
source coding to structured data like trees \cite{KiefferYS09,MagnerTS16,ZhangYK14} and graphs \cite{ChoiS12,Kieffer16}. 
In this paper, we are concerned with universal source coding for trees.
In their recent paper \cite{ZhangYK14}, Kieffer, Yang, and Zhang started to extend their work on grammar-based source coding from
strings to binary trees. For this, they first represent the input tree $t$ by its {\em minimal directed
acyclic graph} $\D_t$ (the minimal DAG of $t$). This is the directed acyclic graph obtained by removing multiple occurrences
of the same subtree from $t$. In a second step, the minimal DAG $\D_t$ is encoded by a binary string $B(\D_t)$; this 
step is similar to the binary coding of SLPs from \cite{KiYa00}. Combining both steps yields a tree encoder $E_{\text{dag}} : \T 
\to \{0,1\}^*$, where $\T$ denotes the set of all binary trees. In order to define universality of such a tree encoder, 
the classical notion of an information source on finite sequences
is replaced in \cite{ZhangYK14} by the notion of  a structured tree source. The precise definition can be found in Section~\ref{sec-trees};
for the moment the reader can think about a collection of probability distributions $(P_n)_{n \in \N}$, where every $P_n$ is a distribution
on a finite non-empty subset $F_n$, and these sets partition $\T$.
The main cases considered in \cite{ZhangYK14} as well as in this paper 
are: (i) $F_n$ is the set of all binary trees with $n$ leaves 
(leaf-centric sources) and (ii) $F_n$ is the set of all binary trees of depth $n$ (depth-centric sources).
Then, the authors of \cite{ZhangYK14} introduce two properties on binary 
tree sources: (i) the domination property (see Section~\ref{sec-dag}, where it is called the weak domination property)
and (ii) the representation ratio negligibility property. The latter states that  
 $\sum_{t \in F_n} P_n(t) \cdot |\D_t|/|t|$  (the average compression ratio achieved by the minimal DAG)
converges to zero for $n \to \infty$, where the size $|t|$ of the binary tree is defined as its number of leaves. The technical main result of 
\cite{ZhangYK14} states that for every structured tree source $(P_n)_{n \in \N}$ satisfying the 
domination property and the representation ratio negligibility property the {\em average case redundancy}
$$
\sum_{t \in F_n, P_n(t)>0} |t|^{-1} \cdot (|E_{\text{dag}}(t)| + \log_2 P_n(t)) \cdot P_n(t)
$$
converges to zero for $n \to \infty$. Finally, two classes of tree sources having the domination 
property and the representation ratio negligibility property are presented in \cite{ZhangYK14}.
One is a class of leaf centric sources, the other one is a class of depth centric sources.
Both sources have the property that every tree with a non-zero probability is balanced in a certain sense, the precise definitions can be found in Section~\ref{sec-leaf-centric} and Section~\ref{sec-depth-centric}.
As a first contribution, we show that for these sources not only the average case redundancy 
but also the worst case redundancy
\begin{equation} \label{eq-worst-case-red}
\max_{t \in F_n, P_n(t)>0} |t|^{-1} \cdot (|E_{\text{dag}}(t)| + \log_2 P_n(t)) 
\end{equation}
converges to zero for $n \to \infty$.
More precisely, we show that \eqref{eq-worst-case-red} is bounded 
by $O(\log \log n / \log n)$ (respectively, $O((\log \log n)^2 / \log n)$) for the presented class of leaf-centric tree sources (respectively, depth-centric tree sources).
To prove this, we use results from \cite{HuLoNo17,Hubschle-Schneider15} according to which the minimal
DAG of a suitably balanced binary tree of size $n$ is bounded by $O(n / \log n)$, respectively $O(n\cdot\log \log n / \log n)$.

Our second main contribution is the application of {\em tree straight-line programs}, briefly TSLPs, for universal tree coding. 
A TSLP is a context-free tree grammar that produces exactly one tree, see Section~\ref{sec-TSLP} for the precise definition
and \cite{Lohrey15dlt} for a survey. TSLPs can
be viewed as the proper generalization of SLPs for trees. Whereas DAGs only have the ability to share repeated subtrees of a tree,
TSLPs can also share repeated tree patterns with a hole (so-called contexts).
In \cite{HuLoNo17}, the authors presented a linear time algorithm that computes
for a given binary tree $t$ of size $n$ a TSLP $\G_t$ of size $O(n / \log n)$. This shows the main advantage of TSLPs over DAGs:
There exist trees of any size $n$ for which the minimal DAG has size $n$ as well. In Section~\ref{sec-binary-coding} we  define
a binary encoding $B$ of TSLPs similar to the ones for SLPs \cite{KiYa00} and DAGs \cite{ZhangYK14}. We then consider the 
combined tree encoder $E_{\text{tslp}} : \T \to \{0,1\}^*$ with $E_{\text{tslp}}(t) = B(\G_t)$, and prove that its worst case redundancy (that is defined as 
in \eqref{eq-worst-case-red} with $E_{\text{dag}}$ replaced by $E_{\text{tslp}}$) is bounded by $O(\log \log n / \log n)$ for every structured
tree source that satisfies the {\em strong domination property} defined in Section~\ref{sec-univesal-tslp}. The strong domination property 
is a strengthening of the domination property from \cite{ZhangYK14}, and this is what we have to pay extra for our TSLP-based encoding
in contrast to the DAG-based encoding from  \cite{ZhangYK14}. On the other hand, our approach has two main advantages over 
\cite{ZhangYK14}:
\begin{itemize}
\item The representation ratio negligibility property from  \cite{ZhangYK14}
is no longer needed.
\item We get bounds on the worst case redundancy instead of the average case redundancy.
\end{itemize}
Both advantages are based on the fact that the grammar-based compressor from \cite{HuLoNo17} computes a TSLP of worst case size $O(n / \log n)$
for a binary tree of size $n$.

Finally, we present a class of leaf-centric sources (Section~\ref{sec-leaf-centric}) as well as a class of depth-centric sources 
 (Section~\ref{sec-depth-centric}) having the strong domination property. These classes are orthogonal to the classes considered
 in \cite{ZhangYK14}.

\section{Preliminaries}

In this section, we introduce some basic definitions concerning
information theory (Section~\ref{sec-empirical}), binary trees 
(Section~\ref{sec-trees}) and tree straight-line programs (Section~\ref{sec-TSLP}). The latter
are our key formalism for the compression of binary trees.
 
With $\N$ we denote the natural numbers including $0$.
We use the standard $O$-notation and if $b$ is a constant, then
we just write $O(\log n)$ for $O(\log_b n)$.
For the unit interval $\{ r \in \mathbb{R} \mid 0 \leq r \leq 1 \}$
we write $[0,1]$.

\subsection{Empirical distributions and empirical entropy} \label{sec-empirical}

Let $\overline{a} = (a_1, a_2, \ldots, a_n)$ be a tuple of elements that are from some (not necessarily finite) set $A$. 
The {\em empirical distribution} $p_{\overline{a}} : \{a_1, a_2, \ldots, a_n \} \to \mathbb{R}$ of $\overline{a}$ is defined by
$$
p_{\overline{a}}(a) = \frac{|\{i \mid 1\le i\le n,\; a_i = a \}|}{n}  .
$$
We use this definition also for words over some alphabet by identifying a word
$w = a_1 a_2 \cdots a_n$ with the tuple $(a_1, a_2, \ldots, a_n)$.
The {\em unnormalized empirical entropy} of $\overline{a}$ is 
$$
H(\overline{a}) = - \sum_{i=1}^n \log p_{\overline{a}}(a_i) .
$$
A well-known generalization of Shannon's inequality states that for 
all real numbers $p_1, \ldots, p_k, q_1, \ldots q_k > 0$, if
$\sum_{i=1}^k p_i = 1 \geq \sum_{i=1}^k q_i$ then
$$
\sum_{i=1}^k -p_i \log_2(p_i) \leq \sum_{i=1}^k -p_i \log_2(q_i) 
$$
see \cite{Aczel} for a proof.
As a consequence, for a tuple $\overline{a} = (a_1, a_2, \ldots, a_n)$ 
with $a_1, \ldots, a_n \in A$ 
and real numbers $q(a) > 0$ ($a \in A$) with $\sum_{i=1}^n q(a_i) \leq 1$ 
we have
\begin{equation} \label{shannon}
\sum_{i=1}^n - \log_2( p_{\overline{a}}(a_i))  \leq \sum_{i=1}^n - \log_2( q(a_i) ) .
\end{equation}

\subsection{Trees, tree sources, and tree compressors} \label{sec-trees}

With $\T$ we denote the set of all binary trees. We identify $\T$ with the set
of terms that are built from the binary symbol $f$ and the constant $a$. Formally,
$\T$ is the smallest set such that (i) $a \in \T$ and (ii) if $t_1, t_2 \in \T$ then also $f(t_1,t_2) \in \T$.
With $|t|$ we denote the number of occurrences of the constant $a$ in $t$. This is the number of leaves
of $t$. Let $\T_n = \{ t \in \T \mid |t| = n\}$ 
 for $n \geq 1$. The depth $d(t)$ of the tree $t$ is recursively defined
by $d(a) = 0$ and $d( f(t_1, t_2)) = \max \{ d(t_1), d(t_2) \}+1$. Let $\T^d = \{ t \in \T \mid d(t) = d \}$ for $d \in \mathbb{N}$.

Occasionally, we will consider a binary tree $t $ as a graph with nodes and edges in the usual way.
Note that a tree $t \in \T_n$ has $2n-1$ nodes in total: $n$ leaves and $n-1$ internal nodes.
For a node $v$ we write $t[v]$ for the {\em subtree} rooted at $v$ in $t$.

A {\em context} is a binary tree $t$, where exactly one leaf is labelled with the special symbol $x$ (called the 
{\em parameter}); all other
leaves are labelled with $a$. For a context $t$ we define $|t|$ to be the number of $a$-labelled leaves of $t$
(which is the number of leaves of $t$ minus 1).
We denote with $\cT$ the set of all contexts and define $\cT_n = \{ t \in \cT \mid |t| = n\}$ for $n \in \mathbb{N}$.
For a tree or context $t\in\T \cup \cT$ and a context $s\in\cT$,
we denote by $s(t)$ the tree or context which results from $s$ by replacing the parameter $x$ by $t$. 
For example $s=f(a,x)$ and $t=f(a,a)$ yields $s(t)=f(a,f(a,a))$.
The depth $d(t)$ of a context $t\in\cT$ is defined as the depth of the tree $t(a)$.

A tree source is a pair $( (\mathcal{F}_i)_{i \in \mathbb{N}}, P)$ such that the following conditions hold:
\begin{itemize}
\item $\mathcal{F}_i \subseteq \T$ is non-empty and finite for every $i \geq 0$,
\item $\mathcal{F}_i \cap \mathcal{F}_j = \emptyset$ for $i \neq j$ and $\bigcup_{i \geq 0} \mathcal{F}_i = \T$, 
i.e., the sets $\mathcal{F}_i$ form a partition of $\T$,
\item $P:\T\to [0,1]$ and $\sum_{t \in \mathcal{F}_i} P(t) = 1$ for every $i \geq 0$, i.e., $P$ restricted to $\mathcal{F}_i$ is a probability
distribution.
\end{itemize}
In this paper, we consider only two cases for the partition $(\mathcal{F}_i)_{i \in \mathbb{N}}$: either 
$\mathcal{F}_i = \T_{i+1}$ for all $i \in \mathbb{N}$ (note that there is no tree of size $0$) or 
$\mathcal{F}_i = \T^i$ for all $i \in \mathbb{N}$. Tree source of the former (resp., latter) type are called {\em leaf-centric}
(resp., {\em depth-centric}).

A {\em tree encoder} is an injective mapping $E : \T \to \{0,1\}^*$ such that the range $E(\T)$ is prefix-free, i.e.,
there do not exist $t, t' \in \T$ with $t \neq t'$ such that $E(t)$ is a prefix of $E(t')$.
We define the {\em worst-case redundancy} of $E$ with respect to 
the tree source $\mathcal{S} = ((\mathcal{F}_i)_{i \in \mathbb{N}}, P)$
as the mapping $i \mapsto R(E, \mathcal{S},i)$ ($i \in \mathbb{N}$) with
$$
R(E, \mathcal{S},i) = \max_{t \in \mathcal{F}_i, P(t)>0}  \frac{1}{|t|} \cdot  (|E(t)| + \log_2 P(t))
$$
The worst-case redundancy is also known as the {\em maximal pointwise redundancy}.

\subsection{Tree straight-line programs} \label{sec-TSLP}

 We now introduce tree straight-line programs for binary trees.
 Let $V$ be a finite ranked alphabet, where each symbol $A\in V$ has an associated rank $0$ or $1$.
 The elements of $V$ are called \emph{nonterminals}.
 We assume that $V$ contains at least one element of rank $0$ and that $V$ is disjoint from the set $\{f,a,x\}$, which are the labels used for binary trees and contexts.
 We use $V_0$ (respectively, $V_1$) for the set of nonterminals of rank $0$ (resp. of rank $1$).
 The idea is that nonterminals from $V_0$ (respectively, $V_1$) derive to trees from $\T$ (respectively, 
 contexts from $\cT$).
 We denote by $\T_V$ the set of trees over $\{f,a\} \cup V$,
i.e. each node in a tree $t\in\T_V$ is labelled with a symbol from $\{f,a\} \cup V$ and
the number of children of a node corresponds to the rank of its label.
With $\cT_V$ we denote the corresponding set of all contexts, i.e., the set of trees
over  $\{f,a,x\} \cup V$, where the parameter symbol $x$ occurs exactly once and at a leaf position.
 Formally, we have $\T \subset \T_V$ and $\cT \subset \cT_V$.
 A \emph{tree straight-line program} $\G$, or short \emph{TSLP}, is a tuple $(V, A_0, r)$,
 where $A_0 \in V_0$ is the start nonterminal and 
$r:V\to (\T_V\cup \cT_V)$ is the function which assigns each nonterminal its unique right-hand side.
It is required that if $A\in V_0$ (respectively, $A \in V_1$), 
then $r(A) \in \T_V$ (respectively, $r(A) \in \cT_V$).
Furthermore, the binary relation $\{ (A,B) \in V \times V \mid  B\text{ is a label in }r(A)\}$ needs to be acyclic.
These conditions ensure that exactly one tree is derived from the start nonterminal $A_0$ 
 by using the rewrite rules $A \to r(A)$ for $A \in V$.
To define this formally, we define $\val_{\G}(t) \in \T$ for $t \in \T_V$ and 
$\val_{\G}(t) \in \cT$ for $t \in \cT_V$ inductively by the following rules:
\begin{itemize}
\item $\val_{\G}(a) = a$ and $\val_{\G}(x)=x$,
\item $\val_{\G}(f(t_1, t_2)) = f( \val_{\G}(t_1), \val_{\G}(t_2))$ for $t_1, t_2 \in  \T_V \cup \cT_V$ (and $t_1 \in \T_V$ or $t_2 \in \T_V$ since there is at most one parameter in $f(t_1,t_2)$),
\item $\val_{\G}(A) = \val_{\G}(r(A))$ for $A \in V_0$,
\item $\val_{\G}(A(s)) = \val_{\G}(r(A)) (\val_{\G}(s) )$ for $A \in V_1$, $s \in  \T_V \cup \cT_V$ 
(note that $\val_{\G}(r(A))$ is a context $t'$, so we can built $t'(\val_{\G}(s) )$).
\end{itemize}
 The tree defined by $\G$ is $\val(\G) = \val_{\G}(A_0)\in\T$. 
 Moreover, for $A \in V_1$ we also write $\val_{\G}(A)$ for $\val_{\G}(A(x))$.

\begin{example}\label{example:TSLP}
	Let $\G = (\{A_0,A_1,A_2\}, A_0, r)$ be a TSLP with $A_0,A_1 \in V_0, A_2 \in V_1$ and
	\[
	r(A_0)=f(A_1,A_2(a)),\; r(A_1) = A_2(A_2(a)), \; r(A_2) = f(x,a).
	\]
	We get $\val_{\G}(A_2) = f(x,a)$, $\val_{\G}(A_1) = f(f(a,a), a)$ and 
	$\val({\G}) = \val_{\G}(A_0) = f( f(f(a,a), a), f(a,a))$.
\end{example}
In this paper, we will consider two classes of syntactically restricted TSLPs:  (i)  DAGs (directed acyclic graphs) and (ii) 
TSLPs in normal form.
Let us start with the former; normal form TSLPs will be introduced in Section~\ref{sec-normal-form}.

\section{Tree compression with DAGs} \label{sec-dag}

In this section we sharpen some of the results from  \cite{ZhangYK14}, where universal source coding of binary trees using
minimal DAGs (directed acyclic graphs) is investigated. In \cite{ZhangYK14}, only bounds on the 
average redundancy for certain classes of tree sources were shown. Here we extend these bounds (for the same classes
of tree sources) to the worst-case redundancy.

\subsection{Directed acyclic graphs (DAGs)}

A DAG is a TSLP $\D = (V, A_0, r)$ such that 
$V = \{ A_0, A_1, \ldots, A_{n-1}\}$ for some $n \in \N$, $n \geq 1$, $V = V_0$ (i.e., all nonterminals have rank $0$), and 
for  every $A_i \in V$, 
the right-hand side $r(A_i)$ is of the form 
$f(\alpha_1,\alpha_2)$ with $\alpha_1, \alpha_2 \in \{a, A_{i+1},\ldots, A_{n-1}\}$. 
Note that a TSLP of this form generates a tree with at least two leaves.
In order to include the tree $a$ with a single leaf, we also 
allow the TSLP $\G_a = (\{A_0\}, A_0, A_0 \mapsto a)$. 
We define the size of a DAG as $|\D| = n+1$. 

In contrast to general TSLPs, every binary tree $t$ has a unique (up to renaming of nonterminals) minimal DAG $\D_t$, 
whose size is the number of different (pairwise non-isomorphic) subtrees of $t$. 
The idea is to introduce for every subtree $f(t_1,t_2)$ of size at least two a nonterminal $A_i$ with 
$r(A_i) = f(\alpha_1,\alpha_2)$, where $\alpha_i = a$ if $t_i = a$ and $\alpha_i$ 
is the nonterminal corresponding to the subtree $t_i$ if $|t_i| \geq 2$ ($i=1,2$).
We will only use this minimal DAG $\D_t$ in the sequel. 

\begin{example} \label{example-dag}
Consider the tree $t_n = f(f(f(\cdots f(a,a), \cdots a),a),a)$, where $f$ occurs $n$ times. We have $|t_n|=n+1$.
The minimal DAG of $t_n$ is 
$( \{ A_0, A_1, \ldots, A_{n-1}\}, A_0, r_n)$, where 
$r_n(A_i) = f(A_{i+1},a)$ for $0 \leq i \leq n-2$ and $r_n(A_{n-1}) = f(a,a)$ and its size is $n+1$.
\end{example}
The above example shows that in the worst-case, the size of the minimal DAG is not smaller than the size of the tree.

\subsection{Universal source coding with DAGs}

The following condition on a tree source  was introduced in  \cite{ZhangYK14},
where it is called the domination property (later, we will introduce a strong domination property):
Let $((\mathcal{F}_i)_{i \in \mathbb{N}}, P)$ be a tree source as defined in Section~\ref{sec-trees}.
We say that $((\mathcal{F}_i)_{i \in \mathbb{N}}, P)$ has 
the {\em weak domination property} if there exists a mapping $\lambda : \T  \to \mathbb{R}_{>0}$
with the following properties:
\begin{enumerate}[(i)]
\item \label{wdom1} $\lambda(t) \geq P(t)$ for every $t \in \T$.
\item \label{wdom2} $\lambda(f(s,t)) \leq \lambda(s) \cdot \lambda(t)$ for all $s,t \in \T$ 
\item \label{wdom4} There are constants $c_1, c_2$ such that $\sum_{t \in \T_n} \lambda(t) \leq c_1 \cdot n^{c_2}$ for all $n \geq 1$.
\end{enumerate}

In \cite{ZhangYK14}, the authors define a binary encoding $B(\D_t) \in \{0,1\}^*$, such that $B(\D_t)$ is not a prefix of 
$B(\D_{t'})$ for all binary trees $t, t'$ with $t \neq t'$.
The precise definition of $B(\D_t)$ is not important for us; all we need is the following
bound from \cite[Theorem~2]{ZhangYK14}, where $E_{\text{dag}} : \T \to \{0,1\}^*$ is the 
tree encoder with $E_{\text{dag}}(t) = B(\D_t)$.

\begin{lemma} \label{lemma-entropy-bound-dag}
Assume that $((\mathcal{F}_i)_{i \in \mathbb{N}}, P)$ has 
the weak domination property. 
Let $t \in \T_n$ with $n \geq 2$ and $P(t)>0$, and let $\D_t$ be the minimal DAG for $t$.
We have
$$
\frac{1}{n} \cdot ( |E_{\text{dag}}(t)| + \log_2( P(t) ) ) \leq  O(|\D_t|/n) + O( |\D_t|/n \cdot \log_2 (n/|\D_t|)) .
$$
\end{lemma}
This bound is used in \cite{ZhangYK14} to show that for certain leaf-centric and depth-centric tree sources 
the encoding $E_{\text{dag}}$ is universal in the sense that the average redundancy converges to zero.
Here, we want to show that for the same tree sources already the worst-case redundancy 
converges to zero. Let us first define the specific classes of tree
sources studied in \cite{ZhangYK14}.

\subsection{Leaf-centric binary tree sources} \label{sec-leaf-centric}

We recall the definition of a natural class of leaf-centric tree sources from  \cite{ZhangYK14}:
Let $\Sigma_{\text{leaf}}$ be the set of all functions $\sigma:(\mathbb{N}\setminus\{0\})\times(\mathbb{N}\setminus\{0\})\to[0,1]$  such that
for all $n\ge 2$:
\begin{equation} \label{eq-sigma-leaf}
\sum_{i,j\ge 1,\;i+j=n}\sigma(i,j)=1 .
\end{equation}
For $\sigma \in \Sigma_{\text{leaf}}$ we define $P_\sigma:\T \to [0,1]$ inductively by:
\begin{eqnarray} \label{lcsource1}
P_\sigma(a) &=& 1, \\    \label{lcsource2}
P_\sigma(f(s,t)) & = & \sigma(|s|, |t|) \cdot P_\sigma(s) \cdot P_\sigma(t) .
\end{eqnarray}
We have $\sum_{t\in\T_n}P_\sigma(t)=1$ and thus 
$( (\T_{i})_{i \geq 1}, P_\sigma)$ is a leaf-centric tree source.
The following result is implicitly shown in \cite{ZhangYK14} (see also the proof of Theorem~\ref{theorem:lcts}).

\begin{lemma} \label{lemma-leaf-centric-weak-dom}
For every $\sigma \in \Sigma_{\text{leaf}}$, the  leaf-centric tree source $( (\T_{i})_{i \geq 1}, P_\sigma)$
has the weak domination property.
\end{lemma}
We say that a mapping $\sigma \in \Sigma_{\text{leaf}}$ is {\em leaf-balanced} if there exists a constant $c$ such that
for all $(i,j) \in (\mathbb{N}\setminus\{0\})\times(\mathbb{N}\setminus\{0\})$ with $\sigma(i,j) > 0$ we have
$$
\frac{i+j}{\min\{i,j\}} \leq c .
$$
In \cite{ZhangYK14} it is shown that for a leaf-balanced $\sigma \in \Sigma_{\text{leaf}}$, the tree source $( (\T_{i})_{i \geq 1}, P_\sigma)$ has
the so called representation ratio negligibility property. This means that the 
average compression ratio achieved by the minimal DAG
(formally, $\sum_{t \in \T_n} P_\sigma(t) \cdot |\D_t|/n$)
converges to zero for $n \to \infty$. Using a result from \cite{HuLoNo17}, we show the following stronger property.

\begin{lemma}\label{lemma:leaf-balanced}
For every leaf-balanced mapping $\sigma \in \in \Sigma_{\text{leaf}}$, there 
exists a constant $\alpha$ such that for every binary tree $t\in \T_n$ with  $P_\sigma(t)>0$
we have  $|\mathcal{D}_t|\leq  \alpha \cdot n/\log_2 n$.
\end{lemma}

\begin{proof}
In~\cite{HuLoNo17} the authors introduce the notion of $\beta$-balanced trees.
Let $v$ be a non-leaf node in $t\in\T$, i.e. $t[v]=f(t_1,t_2)$ for binary trees $t_1$ and $t_2$.
Then $v$ is $\beta$-balanced if $|t_1|\le \beta\cdot|t_2|$ and $|t_2|\le \beta\cdot|t_1|$.
The tree $t$ is $\beta$-balanced if for all non-leaf nodes $u$ and $v$ in $t$ such that $u$ is a child node of $v$,
at least one of the nodes $u$, $v$ is $\beta$-balanced.
It is shown in~\cite{HuLoNo17} that for every constant $\beta$ there exists a constant $\alpha$ (depending only
on $\beta$) such that the minimal dag $\mathcal{D}_t$ of a $\beta$-balanced tree $t\in \T_n$ has size at most $\alpha \cdot n / \log_2 n$.
It follows that we only need to show that a tree $t\in\T_n$ with $P_\sigma(t)>0$ is $\beta$-balanced for a constant $\beta$. 
Since the mapping $\sigma$ is leaf-balanced, every subtree $f(t_1,t_2)$ of $t$ satisfies
$|t_1|+|t_2|\le c\cdot\min\{|t_1|,|t_2|\}$, where $c \geq 1$ is a constant.
Without loss of generality assume that $|t_1|\le|t_2|$. We get $|t_1|\le c\cdot |t_2|$ and
$$|t_2|\le |t_1|+|t_2|\le c\cdot\min\{|t_1|,|t_2|\}=c\cdot |t_1|,$$
which shows that $t$ is $c$-balanced.
\end{proof}

\begin{corollary}
Let $\sigma \in \Sigma_{\text{leaf}}$ be leaf-balanced, and let
$\mathcal{S} = ( (\T_{i})_{i \geq 1}, P_\sigma)$ be the corresponding  leaf-centric tree source.
Then, we have $R(E_{\text{dag}}, \mathcal{S},i) \leq O(\log \log i / \log i)$.
\end{corollary}

\begin{proof}
Let $\alpha$ be the constant from Lemma~\ref{lemma:leaf-balanced}.
Let $t \in \T_n$ such that $P_\sigma(t) > 0$. Lemma~\ref{lemma:leaf-balanced} implies 
$|\mathcal{D}_t|\leq  \alpha \cdot n/\log_2 n$. With Lemma~\ref{lemma-entropy-bound-dag} and
\ref{lemma-leaf-centric-weak-dom} we get
\begin{eqnarray*}
\frac{1}{n} \cdot ( |E_{\text{dag}}(t)| + \log_2( P_\sigma(t) ) ) & \leq &  O(|\D_t|/n) + O( |\D_t|/n \cdot \log_2 (n/|\D_t|)) \\
& \leq & O(\log \log n / \log n)
\end{eqnarray*}
This proves the corollary.
\end{proof}

\subsection{Depth-centric binary tree sources} \label{sec-depth-centric}

We recall the definition of a natural class of depth-centric tree sources from~\cite{ZhangYK14}:
Let $\Sigma_{\text{depth}}$ be the set of all mappings 
$\sigma:\mathbb{N}\times\mathbb{N}\to[0,1]$ such that for all $n\ge 1$:
\begin{equation} \label{eq-sigma-depth}
	\sum_{i,j\ge 0,\;\max(i,j)=n-1}\sigma(i,j)=1 .
\end{equation} 
For $\sigma \in \Sigma_{\text{depth}}$, we define $P_\sigma:\T \to [0,1]$ by 
\begin{eqnarray} \label{dcsource1}
P_\sigma(a) &=& 1, \\ \label{dcsource2}
P_\sigma(f(s,t)) &=& \sigma(d(s),d(t)) \cdot P_\sigma(s) \cdot P_\sigma(t) .
\end{eqnarray}
We have $\sum_{t\in\T^n}P_\sigma(t)=1$ and thus 
$( (\T^{i})_{i \geq 0}, P_\sigma)$ is a depth-centric tree source.

The following result is again implicitly shown in~\cite{ZhangYK14} and can also be found as a part of the proof of Theorem~\ref{theorem:dcts}.
\begin{lemma} \label{lemma-depth-centric-weak-dom}
For every $\sigma \in \Sigma_{\text{depth}}$, the  depth-centric tree source $( (\T_{i})_{i \geq 1}, P_\sigma)$
has the weak domination property.
\end{lemma}

We say the mapping $\sigma \in \Sigma_{\text{depth}}$ is {\em depth-balanced}
if there exists a constant $c$ such that
for all $(i,j) \in \mathbb{N}\times\mathbb{N}$ with $\sigma(i,j) > 0$ we have $|i-j| \leq c$. 
In~\cite{ZhangYK14}, the authors define a condition on $\sigma$ that is slightly stronger than depth-balancedness,
and show that for every such $\sigma$, the tree source $( (\T^{i})_{i \geq 0}, P_\sigma)$ has the representation ratio negligibility property.
Similarly to Lemma~\ref{lemma:leaf-balanced}, we will show an even stronger property.
To do so, we introduce {\em $\beta$-depth-balanced} trees for $\beta \in \N$.
A tree $t$ is called {\em $\beta$-depth-balanced} if for each 
subtree $f(t_1,t_2)$ of $t$ we have $|d(t_1)-d(t_2)|\le\beta$.
Note that for a depth-balanced mapping $\sigma \in \Sigma_{\text{depth}}$,
there is a constant $\beta$ such that every tree $t$ with $P_\sigma(t)>0$ is $\beta$-depth-balanced.
We will use the following lemma:

\begin{lemma}\label{lemma:beta-depth-balanced}
Let $\beta\in\mathbb{N}$ and $c = 1+1/(1+\beta)$ (thus, $1 < c \leq 2$). For every $\beta$-depth-balanced binary tree $t$, we have
$|t|\ge c^{d(t)}$.
\end{lemma}

\begin{proof}
We prove the lemma by induction on $d(t)$.
For the only tree $t=a$ of depth $d(t)=0$, we have $|t|=1=c^0$.
Consider now a $\beta$-depth-balanced tree $t=f(t_1,t_2)$ of depth $d(t)>0$. 
We assume $d(t_1)\ge d(t_2)$, the other case is symmetric.
Since $t$ is $\beta$-depth-balanced, it follows that $d(t_2)\ge d(t_1)-\beta$.
To estimate the size $|t|=|t_1|+|t_2|$, we apply the induction hypothesis to $t_1$ and $t_2$, which yields
$$
|t| = |t_1|+|t_2|\ge {c}^{d(t_1)}+c^{d(t_2)}\ge c^{d(t_1)}+{c}^{d(t_1)-\beta}= c^{d(t_1)}\cdot\left(1+c^{-\beta}\right).
$$
Since $d(t_1)+1=d(t)$, it only remains to show that $1+{c}^{-\beta} \ge c$, which can be easily done by induction on $\beta\in\mathbb{N}$.
\end{proof}
Lemma~\ref{lemma:beta-depth-balanced} together with results from \cite{HuLoNo17,Hubschle-Schneider15} implies:

\begin{lemma}\label{lemma:depth-balanced}
For every depth-balanced mapping $\sigma \in \Sigma_{\text{depth}}$ there 
exists a constant $\alpha$ such that for every binary tree $t\in \T_n$ with $P_\sigma(t)>0$
we have  $|\mathcal{D}_t|\leq  \alpha \cdot n \cdot \log_2 (\log_2 n) /\log_2 n$.
\end{lemma}

\begin{proof}
If $P_\sigma(t) > 0$, then there exists a constant $\beta$ such that $t$ and each of its subtrees is $\beta$-depth-balanced.
By Lemma~\ref{lemma:beta-depth-balanced} this implies that every subtree $t'$ has depth at most $c \cdot \log_2 |t'|$ for a constant
$c$ that only depends on $\sigma$. By \cite[Theorem~12]{HuLoNo17} 
(which was implicitly shown in \cite{Hubschle-Schneider15}) it follows that there exists a constant $\alpha$ (again, only dependent on $\sigma$)
such that $|\mathcal{D}_t|\leq  \alpha \cdot n \cdot \log_2 (\log_2 n) /\log_2 n$.
\end{proof}

\begin{corollary}
Let $\sigma \in \Sigma_{\text{depth}}$ be a depth-balanced mapping and let
$\mathcal{S} = ( (\T^{i})_{i \geq 0}, P_\sigma)$ be the corresponding  leaf-centric tree source.
Then, we have $R(E_{\text{dag}}, \mathcal{S},  i) \leq O( (\log \log i)^2 / \log i)$.
\end{corollary}

\begin{proof}
Let $\alpha$ be the constant from Lemma~\ref{lemma:depth-balanced}.
Let $t \in \T^i$ such that $P_\sigma(t) > 0$. 
Lemma~\ref{lemma-entropy-bound-dag} and \ref{lemma-depth-centric-weak-dom} imply
$$
\frac{1}{|t|} \cdot ( |E_{\text{dag}}(t)| + \log_2( P_\sigma(t) ) ) \leq  O(|\D_t|/|t|) + O( |\D_t|/|t| \cdot \log_2 (|t|/|\D_t|)) .
$$
Consider the mapping $g(x) = x \cdot \log_2(1/x)$. It is monotonically increasing for $0 \leq x \leq 1/e$.
Note that for all $t \in \T^i$ we have $|t| \geq i+1$. 
Hence, if $i$ is large enough, then Lemma~\ref{lemma:depth-balanced} yields for all $t \in \T^i$ with $P_\sigma(t) > 0$ that
$$
|\D_t|/|t| \leq  \alpha \cdot \log_2 (\log_2 |t|) /\log_2 |t| \leq \log_2 (\log_2 (i+1)) /\log_2 (i+1) \leq 1/e.
$$
We obtain
$$
\frac{1}{|t|} \cdot ( |E_{\text{dag}}(t)| + \log_2( P_\sigma(t) ) ) \leq O((\log \log i)^2 / \log i) .
$$
This proves the corollary.
\end{proof}

\section{Tree compression with TSLPs}

In this section, we will use general TSLPs for the compression of binary trees.
The limitations of DAGs for universal source coding can be best seen for a tree
source $( (\mathcal{F}_i)_{i \in \mathbb{N}}, P)$ such that $P(t) > 0$ for all $t \in \T$.
Example~\ref{example-dag} shows that for every $n \geq 1$, there is a tree $t \in \T_n$ with
$|\mathcal{D}_t| = n$.  In that case, the bound stated in Lemma~\ref{lemma-entropy-bound-dag} cannot
be used to show that the worst-case redundancy converges to zero.

\subsection{TSLPs in normal form} \label{sec-normal-form}

In this section, we will use TSLPs in a certain normal form, which we introduce first.

A TSLP $\G = (V, A_0, r)$ is in {\em normal form} if the following conditions hold:
\begin{itemize}
\item $V = \{ A_0, A_1, \ldots, A_{n-1}\}$ for some $n \in \N$, $n \geq 1$.
\item For every $A_i \in V_0$, 
the right-hand side $r(A_i)$ is a term of the form 
$A_j(\alpha)$, where $A_j \in V_1$ and $\alpha \in V_0 \cup \{a\}$.
\item For every $A_i \in V_1$ 
the right-hand side $r(A_i)$ is a term of the form 
$A_j(A_k(x))$, $f(\alpha,x)$, or $f(x,\alpha)$,
where $A_j,A_k \in V_1$ and $\alpha \in V_0 \cup \{a\}$.  
\item For every $A_i \in V$ define the word $\rho(A_i) \in (V \cup \{a\})^*$ as follows:
$$
\rho(A_i) = \begin{cases}
A_j \alpha & \text{ if } r(A_i) = A_j(\alpha) \\
A_j A_k & \text{ if } r(A_i) = A_j(A_k(x)) \\
\alpha & \text{ if } r(A_i) = f(\alpha,x) \text{ or } f(x,\alpha)
\end{cases}
$$
Let $\rho_{\G} = \rho(A_0) \rho(A_1) \cdots \rho(A_{n-1}) \in \{a,A_1,A_2,\ldots, A_{n-1}\}^*$. 
Then we require that $\rho_{\G}$ is of the form 
$\rho_{\G} = A_1 u_1 A_2 u_2 \cdots A_{n-1}  u_{n-1}$ 
with $u_i \in \{a, A_1,  A_2,\ldots, A_i \}^*$.
\item $\val_{\G}(A_i) \neq \val_{\G}(A_j)$ for $i \neq j$
\end{itemize}
As for DAGs we also allow the TSLP $\G_a = (\{A_0\}, A_0, A_0 \mapsto a)$ in order to get the singleton tree $a$.
In this case, we set $\rho_{\G_a} = \rho(A_0) =  a$.

Let $\G = (V,A_0,r)$ be a TSLP in normal form with $V = \{ A_0, A_1, \ldots, A_{n-1}\}$ for the further definitions.
We define the size of $\G$ as $|\G| = |\rho_{\G}|$.
This is the total number of occurrences of symbols
from $V \cup \{a\}$ in all right-hand sides of $\G$.
Let $\omega_{\G}$ be the word obtained from $\rho_{\G}$ by removing for every $1\le i \le n-1$ the first
occurrence of $A_i$ from $\rho_{\G}$. Thus, if $\rho_{\G} = A_1 u_1 A_2 u_2 \cdots A_{n-1}  u_{n-1}$ 
with $u_i \in \{a, A_1,  A_2,\ldots, A_i \}^*$, then $\omega_{\G} = u_1 u_2 \cdots u_{n-1}$.
The {\em entropy} $H(\G)$ of the normal form TSLP $\G$ is defined as the empirical unnormalized entropy of the word $\omega_{\G}$:
$$
H(\G) = H(\omega_{\G}) .
$$
\begin{example}\label{example:TSLPnormalform}
Let $\G=(\{A_0,A_1,A_2,A_3,A_4\},A_0,r)$ be the normal form TSLP with $A_0,A_2,A_3 \in V_0, A_1,A_4 \in V_1$ and
\begin{eqnarray*}
&&r(A_0)=A_1(A_2),\; r(A_1) = f(x,A_3), \; r(A_2) = A_4(A_3),\\
&&r(A_3)=A_4(a), \; r(A_4)=f(x,a).
\end{eqnarray*}
We have $\val(\G)=f( f(f(a,a), a), f(a,a))$, $\rho_{\G}=A_1A_2A_3A_4A_3A_4aa$ ($u_1=u_2=u_3=\varepsilon$, $u_4=A_3A_4aa$), $|\G|=8$ and $\omega_\G=A_3A_4aa$. 
\end{example}
The {\em derivation tree} $T_{\G}$ of $\G$ is a rooted tree, where every node is labelled with a symbol
from $V \cup \{a\}$.
The root is labelled with $A_0$. 
Nodes labelled with $a$ are the leaves of $\T_G$.
A node $v$ that is labelled with a nonterminal $A_i$ has $|\rho(A_i)|$ many children.
If $\rho(A_i) = \alpha \in V_0 \cup \{a\}$ then the single child of $v$ is labelled with $\alpha$.
If $\rho(A_i) = A_j \alpha$ with $\alpha \in V \cup \{a\}$ then the left (resp., right) child
of $v$ is labelled with $A_j$ (resp., $\alpha$). 
An {\em initial subtree} of $T_{\G}$ is a tree that can be obtained from $T_{\G}$ as follows:
Take a subset $U$ of the nodes of $T_{\G}$ and remove from $T_{\G}$ all proper descendants of nodes from $U$, i.e., all nodes that are located strictly below a node from $U$. 
For such an initial subtree we denote with $\val_{\G}(T') \in \T_V$ the tree that is derived from the start nonterminal 
$A_0$ by using the rules according to $T'$. Formally, it is defined by assigning to every node $v$ of $T'$ a tree or context
$t_v \in \T_V \cup \cT_V$ inductively as follows:
\begin{itemize}
\item If $v$ is labelled with $a$ then $t_v = a$.
\item If $v$ is labelled with $A_i \in V_0$ and is a leaf of $T'$, then $t_v = A_i$.
\item If $v$ is labelled with $A_i \in V_1$ and is a leaf of $T'$, then $t_v = A_i(x)$.
\item If $v$ is labelled with $A_i \in V_1$ and has a single child node $u$, then $t_v = r(A_i)(t_u)$
(note that $r(A_i)$ must be of the form $f(x,\alpha)$ or $f(\alpha,x)$).
\item If $v$ has the left child $u_1$ and the right child $u_2$, then $t_v = t_{u_1}(t_{u_2})$.
\end{itemize}
We finally set $\val_{\G}(T') = t_{v_0}$, where $v_0$ is the root node of $T'$.

\begin{example} \label{example-derivation-tree}
Let $\G$ be the normal form TSLP from Example~\ref{example:TSLPnormalform}.
The derivation tree $T_{\G}$ is shown in Figure~\ref{fig-derivation-tree} on the left; an initial subtree $T'$ of it is shown on the right.
We have $\val_{\G}(T') =f(A_4(A_3), f(a,a))$.
\end{example}

\begin{figure}[t]
		\tikzset{level 1/.style={sibling distance=24mm}}
		\tikzset{level 2/.style={sibling distance=16mm}}
		\tikzset{level 7/.style={sibling distance=8mm}}  
		\tikzset{level 8/.style={sibling distance=4mm}}
		\hspace*{\fill}
		\begin{tikzpicture}[scale=1,auto,swap,level distance=8mm]
		\node (eps) {$A_0$} 
		child {node {$A_1$}
			child {node {$A_3$}
				child {node {$A_4$}
					child {node{$a$}}
				}
				child {node {$a$}}
			}
		}
		child {node {$A_2$}
			child {node {$A_4$}
				child {node {$a$}}
			}
			child {node {$A_3$}
				child {node {$A_4$}
					child {node{$a$}}
				}
				child {node {$a$}}
			}
		}
		;
		{label fig one};
		\end{tikzpicture}
		\hspace*{\fill}
		\begin{tikzpicture}[scale=1,auto,swap,level distance=8mm]
		\node (eps) {$A_0$} 
		child {node {$A_1$}
			child {node {$A_3$}
				child {node {$A_4$}
					child {node{$a$}}
				}
				child {node {$a$}}
			}
		}
		child {node {$A_2$}
			child {node {$A_4$}
			}
			child {node {$A_3$}
			}
		}
		;
		{label fig one};
		\end{tikzpicture}
		\hspace*{\fill}
		\caption{The derivation tree $T_\G$ of the TSLP from Example~\ref{example-derivation-tree} (left) and an initial subtree $T'$ of $T_\G$ (right).}
		\label{fig-derivation-tree}
	\end{figure}
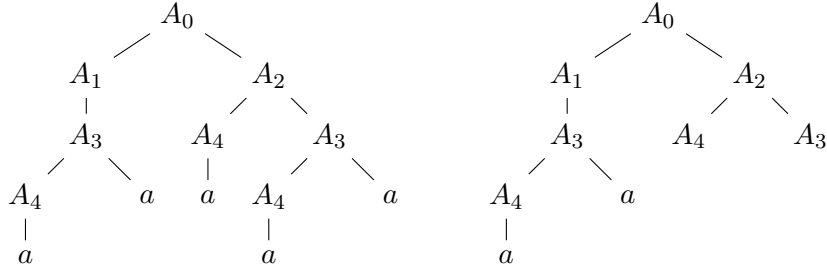

A {\em grammar-based tree compressor} is an algorithm $\psi$ that produces for a given tree 
$t \in \T$ a TSLP $\mathcal{G}_t$ in normal form. The {\em compression ratio} of $\psi$ is 
the mapping $n \mapsto \gamma_\psi(n)$ with
$$
\gamma_\psi(n) = \max_{t \in \T_{n}} |\mathcal{G}_t|/n.
$$
It is not hard to show that every TSLP can be transformed with a linear size increase into a normal form
TSLP that derives the same tree. For example, the TSLP from Example~\ref{example:TSLP} is transformed into the normal form TSLP described in Example~\ref{example:TSLPnormalform}.
We will not use this fact, since all we need is the following theorem from  \cite{HuLoNo17}:

\begin{theorem} \label{thm-compression-ratio}
There exists a grammar-based compressor $\psi$ (working in linear  time) with 
 $\gamma_\psi(n)  \in O(1 / \log n)$.
\end{theorem}

\subsection{Binary coding of TSLPs in normal form} \label{sec-binary-coding}

In this section we fix a binary encoding for normal form TSLPs. This encoding
is similar to the one for SLPs \cite{KiYa00} and DAGs \cite{ZhangYK14}.
Let $\G = (V,A_0,r)$ be a TSLP in normal form with $n = |V|$ nonterminals.
Let $m = |\G| = |\rho_{\G}|$ be the size of $\G$.  
We define the type $\type(A_i) \in \{0,1,2,3\}$ of a nonterminal $A_i \in V$ 
as follows: 
$$
\type(A_i) = \begin{cases}
0 & \text{ if } \rho(A_i) \in V_1 (V_0 \cup \{a\}) \\
1 & \text{ if } \rho(A_i) \in V_1 V_1 \\
2 & \text{ if } \rho(A_i) = f(\alpha,x) \text{ for some } \alpha \in V_0 \cup \{a\} \\
3 & \text{ if } \rho(A_i) = f(x,\alpha) \text{ for some } \alpha \in V_0 \cup \{a\} 
\end{cases}
$$
We define the binary word $B(\G) = w_0 w_1 w_2 w_3 w_4$, where the words $w_i  \in \{0,1\}^+$, $0\le i \le 4$, are defined as follows:
\begin{itemize}
\item $w_0 = 0^{n-1}1$,
\item $w_1 = a_0 b_0 a_1 b_1 \cdots a_{n-1} b_{n-1}$, where $a_j b_j$ 
is the 2-bit binary encoding of $\type(A_j)$,
\item Let $\rho_{\G} = A_1 u_1 A_2 u_2 \cdots A_{n-1}  u_{n-1}$ with $u_i \in \{a, A_1,  A_2,\ldots, A_i \}^*$.
Then $w_2 = 1 0^{|u_1|} 1 0^{|u_2|} \cdots 1 0^{|u_{n-1}|}$. 
\item For $1\le i \le n-1$ let $k_i = |\rho_{\G}|_{A_i} \geq 1$ be the number of occurrences of the nonterminal $A_i$
in the word $\rho_{\G}$. Then $w_3 = 0^{k_1-1} 1 0^{k_2-1} 1 \cdots 0^{k_{n-1}-1} 1$. 
\item The word $w_4$ encodes the word $\omega_{\G}$ using the well-known enumerative encoding \cite{Cover73}.
Every nonterminal $A_i$, $1\le i\le n-1$, has $\eta(A_i) := k_i-1$ occurrences in $\omega_{\G}$.  The 
symbol $a$ has $\eta(a) := m - (k_1+ \cdots + k_{n-1})$ many occurrences in $\omega_{\G}$.
Let $S$ be the set of words over the alphabet $\{a,A_1,\ldots,A_{n-1}\}$ with $\eta(a)$ occurrences of $a$ and 
$\eta(A_i)$ occurrences of $A_i$ for every $1\le i \le n-1$. Hence, 
\begin{equation} \label{size-S}
|S| = \frac{(m-n+1)!}{ \eta(a)! \cdot \prod_{i=1}^{n-1} \eta(A_i)!} .
\end{equation}
Let $v_0, v_1, \ldots, v_{|S|-1}$ be the lexicographic enumeration of the words from $S$ with respect to 
the alphabet order $a, A_1,\ldots,A_{n-1}$.
Then $w_4$ is the binary encoding of the unique index $i$ such that $\omega_{\G} = v_i$, where
$|w_4| = \lceil \log_2 |S| \rceil$ (leading zeros are added to the binary encoding of $i$ to obtain the length
$ \lceil \log_2 |S| \rceil$).
\end{itemize}

\begin{example}
Consider the normal from TSLP $\G$ from Example~\ref{example:TSLPnormalform}.
We have
$w_0=00001$,
$w_1 = 00 11 00 00 11$,
$w_2=11110000$ and
$w_3=110101$.
To compute $w_4$, note first that there are $|S|=12$ words with two occurrences of $a$ and one occurrence of $A_3$ and $A_4$.
It follows that $|w_4|=\lceil\log_2(12)\rceil=4$.
Further, since the order of the alphabet is $a,A_3,A_4$,
there are only three words in $S$ ($A_4A_3aa$, $A_4aA_3a$ and $A_4aaA_3$), which are lexicographically larger than $\omega_\G=A_3A_4aa$.
Hence, $\omega_\G=v_{8}$ and thus $w_4=1000$.
\end{example}

\begin{lemma}
The set of code words $B(\G)$, where $\G$ ranges over all TSLPs in normal form, is a prefix code.
\end{lemma}

\begin{proof}
Let $B(\G) = w_0 w_1 w_2 w_3 w_4$ with $w_i$ defined as above. We show how to recover the TSLP $\G$.
From $w_0$ we can determine $n = |V|$ and the factor $w_1$ of 
$B(\G)$. Hence, we can determine the type of every nonterminal. Using the types, we can compute $|\G| = |\rho_{\G}|$. 
Moreover, the types allow to compute $\G$ from the word $\rho_G$. Hence, it remains to determine $\rho_{\G}$.
From $|\G|$ one can compute $w_2$. To compute $\rho_\G$ from $w_2$, one only needs $\omega_{\G}$. 
For this, one determines the frequencies $\eta(a), \eta(A_1), \ldots, \eta(A_{n-1})$ 
of the symbols in $\omega_{\G}$ from $w_3$. Using these frequencies one computes the size $|S|$ from \eqref{size-S}
and the length $\lceil \log_2 |S| \rceil$ of $w_4$. From $w_4$, one can finally compute $\omega_G$.
\end{proof}

Note that $|B(\G)| \leq O(|\G|) + |w_4|$. By using the well-known bound on the code length 
of enumerative encoding \cite[Theorem~11.1.3]{CoTh06}, we get:
  
\begin{lemma} \label{lemma-binary-coding}
For the length of the binary coding $B(\G)$ we have: 
$|B(\G)| \leq O(|\G|) + H(\G)$.  
\end{lemma}

\subsection{Universal source coding based on TSLPs in normal form} \label{sec-univesal-tslp}

Let $((\mathcal{F}_i)_{i \in \mathbb{N}}, P)$ be a tree source as defined in Section~\ref{sec-trees}.
We say that $((\mathcal{F}_i)_{i \in \mathbb{N}}, P)$ has 
the {\em strong domination property} if the there exists a mapping $\lambda : \T \cup \cT \to \mathbb{R}_{>0}$
with the following properties:
\begin{enumerate}[(i)]
\item \label{dom1} $\lambda(t) \geq P(t)$ for every $t \in \T$.
\item \label{dom2} $\lambda(f(s,t)) \leq \lambda(s) \cdot \lambda(t)$ for all $s,t \in \T$ 
\item \label{dom3} $\lambda(s(t)) \leq \lambda(s) \cdot \lambda(t)$ for all $s \in \cT$ and $t \in \T$ 
\item \label{dom4} There are constants $c_1, c_2$ such that $\sum_{t \in \T_n \cup \cT_n} \lambda(t) \leq c_1 \cdot n^{c_2}$ for all $n \geq 1$.
\end{enumerate}
The proof of the following lemma combines ideas from \cite{KiYa00} and \cite{ZhangYK14}.

\begin{lemma} \label{lemma-entropy-bound}
Assume that $((\mathcal{F}_i)_{i \in \mathbb{N}}, P)$ has 
the strong domination property. 
Let $t \in \T_n$ with $n \geq 2$ and $P(t)>0$, and let $\G = (V,A_0,r)$ be a TSLP in normal form with $\val(\G)=t$.
We have
$$
H(G) \leq  - \log_2 P(t) + O(|\G|) + O( |\G| \cdot \log_2 (n/|\G|)).
$$
\end{lemma}

\begin{proof}
Let $m = |\G| = |\rho_{\G}|$ be the size of $\G$, $k = |V|$, and $\ell := m+1-k \leq m$. 
Let $T = T_{\G}$ be the derivation tree of $\G$. We define an initial subtree $T'$ as follows: If $v_1$ and $v_2$
are non-leaf nodes of $T$ that are labelled with the same nonterminal and $v_1$ comes before $v_2$
in preorder, then we remove from $T$ all proper descendants of $v_2$. Thus, for every $A_i \in V$ there 
is exactly one non-leaf node in $T'$ that is labelled with $A_i$. Let $t' = \val_{\G}(T')$. In the derivation of $t'$ from $A_0$, every rule of $\G$ is used exactly
once. For the TSLP from Example~\ref{example:TSLPnormalform}, the tree $T'$ is shown in Figure~\ref{fig-derivation-tree} on the right,
and the tree $t'$ is computed in Example~\ref{example-derivation-tree}.

Note that $T'$ has exactly $m+1$ many nodes and $k$ non-leaf nodes. 
Thus, $T'$ has $\ell$ leaves.
Let $v_1, v_2, \ldots, v_{\ell}$ be the sequence of all leaves of $T'$ (w.l.o.g. in preorder)
and let $\alpha_i \in  \{a,A_1,\ldots,A_{k-1}\}$ be the label of $v_i$.
Let $\overline{\alpha} = (\alpha_1, \alpha_2, \ldots, \alpha_{\ell})$.
Then $|\omega_{\G}|_\alpha = |\overline{\alpha}|_\alpha$ for every $\alpha \in \{a,A_1,\ldots,A_{k-1}\}$,
and this is also the number of occurrences of $\alpha$ in the tree $t'$.
Hence, $p_{\overline{\alpha}}$ and $p_{\omega_{\G}}$ have the same empirical distribution. 
Let $t_i = \val_{\G}(\alpha_i) \in \T \cup \cT$.
For the TSLP from Example~\ref{example:TSLPnormalform} we get 
$\overline{\alpha} = (a,a,A_4,A_3)$.
Since $\val_{\G}(A_i) \neq \val_{\G}(A_j)$ for all $i \neq j$ and $\val_{\G}(A_i) \neq a$ for all $i$ (this holds for every
normal form TSLP that produces a tree of size at least two), 
the tuple $\overline{t} = (t_1, t_2, \ldots, t_{\ell})$ satisfies
$p_{\omega_{\G}}(\alpha_i) = p_{\overline{t}}(t_i)$ for all $1\le i\le \ell$.

Using conditions \eqref{dom2} and \eqref{dom3} of the strong domination property, we get
\begin{equation} \label{ineq-lambda}
\lambda(t) \leq \prod_{i=1}^{\ell} \lambda(t_i) .
\end{equation}
For $j \in \N$, $j \geq 1$, let 
$$
M_j = \sum_{u \in \T_j \cup \cT_j} \lambda(u) \leq c_1 \cdot j^{c_2},
$$
where $c_1$ and $c_2$ are the constants from condition \eqref{dom4} of the strong domination property.
Let $D := 6/\pi^2 \geq 1/2$ and define for every $u \in \T_j \cup \cT_j$:
\begin{equation} \label{def-q(u)}
q(u) := \frac{D \cdot \lambda(u)}{M_j \cdot j^2} > 0.
\end{equation}
We get 
$$
\sum_{j \geq 1} \sum_{u \in \T_j \cup \cT_j}  q(u) = D \cdot \sum_{j \geq 1} \frac{1}{j^2} = 1 .
$$
Hence, we have $q(t_1) + q(t_2) + \cdots + q(t_{\ell}) \leq 1$. 
Using Shannon's inequality \eqref{shannon} we get
\begin{equation*}
H(G) = H(\omega_{\G}) =  \sum_{i=1}^{\ell} -\log_2 p_{\omega_{\G}}(\alpha_i) =   \sum_{i=1}^{\ell}  -\log_2 p_{\overline{t}}(t_i) \leq  \sum_{i=1}^{\ell}  -\log_2 q(t_i) .
\end{equation*}
Using \eqref{def-q(u)} and $D \geq 1/2$ we obtain
\begin{eqnarray*}
H(G) & \leq &  \sum_{i=1}^{\ell}  -\log_2 \bigg( \frac{D \cdot \lambda(t_i)}{M_{|t_i|} \cdot |t_i|^2} \bigg) \\
& = & -\ell \cdot \log_2 D  - \sum_{i=1}^{\ell} \log_2 \lambda(t_i) + \sum_{i=1}^{\ell} \log_2 M_{|t_i|} + 2 \sum_{i=1}^{\ell}  \log_2 |t_i| \\
& \leq & \ell  - \sum_{i=1}^{\ell} \log_2 \lambda(t_i) + \sum_{i=1}^{\ell} (\log_2 c_1 + c_2 \log_2 |t_i|) + 2 \sum_{i=1}^{\ell}  \log_2 |t_i|  \\
& = & (1+\log_2 c_1) \cdot \ell - \sum_{i=1}^{\ell} \log_2 \lambda(t_i) + (2+c_2) \cdot \sum_{i=1}^{\ell}  \log_2 |t_i| .
\end{eqnarray*}
From \eqref{ineq-lambda} and condition \eqref{dom1} of the strong domination property we get
$$
 \sum_{i=1}^{\ell} \log_2 \lambda(t_i)  \geq \log_2 \lambda(t) \geq \log_2 P(t) .
$$
Moreover, Jensen's inequality gives
$$
\sum_{i=1}^{\ell}  \log_2 |t_i|  \leq \ell \cdot \log_2 \bigg( \frac{1}{\ell} \cdot \sum_{i=1}^{\ell} |t_i| \bigg) =  \ell \cdot \log_2 (n/\ell).
$$
With $\ell \leq m$ we obtain 
\begin{eqnarray*}
H(G) & \leq & (1+\log_2 c_1) \cdot \ell - \log_2 P(t) + (2+c_2) \cdot  \ell \cdot \log_2 (n/\ell) \\
& \leq &   - \log_2 P(t) + (1+\log_2 c_1) \cdot m +  (2+c_2) \cdot  m \cdot \log_2 (n/m).
\end{eqnarray*}
This shows the lemma.
\end{proof}

Let $\psi: t \mapsto \G_t$ be a grammar-based tree compressor. 
We then consider the tree encoder $E_\psi : \T \to \{0,1\}^*$ defined
by $E_\psi(t) = B(\G_t)$. 
Recall the definition of the  worst-case redundancy $R(E_\psi, \mathcal{S},i)$
from Section~\ref{sec-trees}.

\begin{theorem} \label{thm-red-bound}
Assume that $\mathcal{S} = ((\mathcal{F}_i)_{i \in \mathbb{N}}, P)$ has 
the strong domination property. Let $\psi$ be a grammar-based compressor such that
$\gamma_\psi(n) \leq \gamma(n)$ for  a monotonically decreasing function $\gamma(n)$ with 
$\lim_{n \to \infty} \gamma(n) = 0$. Let $n_i = \min \{ |t| \mid t \in \mathcal{F}_i \}$ 
and assume that $n_i < n_{i+1}$ for all $i \in \mathbb{N}$.\footnote{This is the case
for leaf-centric and depth-centric tree sources}
Then, we have
$$
R(E_\psi, \mathcal{S},i) \leq O(\gamma(n_i) \cdot \log_2 (1/\gamma(n_i))).
$$
\end{theorem}

\begin{proof}
Let $t \in \T$ and $\psi(t) = \G_t$.  
With Lemma~\ref{lemma-binary-coding} and \ref{lemma-entropy-bound} we get
\begin{eqnarray*}
 \frac{1}{|t|} \cdot  (|B(\G_t)| + \log_2 P(t))  & \leq & \frac{1}{|t|} \cdot  (H(\G_t) + O(|\G_t|) + \log_2 P(t)) \\
 & \leq &  \frac{1}{|t|} \cdot  ( O(|\G_t|) + O( |\G_t| \cdot \log_2 (|t|/|\G_t|)) \\
 & = & O(|\G_t|/|t|) + O( |\G_t|/|t| \cdot \log_2 (|t|/|\G_t|)) .
\end{eqnarray*}
Consider the mapping $g$ with $g(x) = x \cdot \log_2(1/x)$. It is monotonically increasing for  $0 \leq x \leq 1/e$. 
If $i$ is large enough, we have for all $t \in \mathcal{F}_i$ that
$$|\G_t|/|t| \leq \gamma_\psi(|t|) \leq \gamma(|t|) \leq  \gamma(n_i) \leq 1/e . 
$$
Hence, we get
$$
|\G_t|/|t| \cdot \log_2 (|t|/|\G_t|) = g(|\G_t|/|t|) \leq g( \gamma(n_i)) = \gamma(n_i) \cdot \log_2 (1/\gamma(n_i)) .
$$
This implies 
\begin{eqnarray*}
R(E_\psi, \mathcal{S},i) & = & \max_{t \in  \mathcal{F}_i, P(t) > 0} \frac{1}{|t|} \cdot  (|B(\G_t)| + \log_2 P(t)) \\
& \leq  & O(\gamma(n_i)) + O( \gamma(n_i) \cdot \log (1/\gamma(n_i))) \\
& = & O( \gamma(n_i) \cdot \log (1/\gamma(n_i))) ,
\end{eqnarray*}
which proves the theorem.
\end{proof}

Note that the minimal size of a tree in $\T_{i+1}$ (resp. $\T^i$) is $i+1$.
Hence, Theorem~\ref{thm-compression-ratio} and \ref{thm-red-bound} yield:

\begin{corollary} \label{coro1}
There exists a grammar-based tree compressor $\psi$ (working in linear time) such that 
$R(E_\psi, \mathcal{S},i) \leq O( \log\log i / \log i)$ for every leaf-centric or depth-centric tree source 
$\mathcal{S}$ having the strong domination property.
\end{corollary}
In the rest of the paper, we will present classes of leaf-centric and depth-centric tree sources
that have the strong domination property.

\subsection{Leaf-centric binary tree sources} \label{sec-leaf-centric2}

Recall the definition of the class of mappings $\Sigma_{\text{leaf}}$ by equations
\eqref{eq-sigma-leaf}, \eqref{lcsource1}, and \eqref{lcsource2} in Section~\ref{sec-leaf-centric} and
the corresponding class of leaf-centric tree sources.
In this section, we state a condition on the mapping $\sigma \in \Sigma_{\text{leaf}}$ that enforces the strong domination
property for the leaf-centric tree source $( (\T_{i})_{i \geq 1}, P_\sigma)$. This allows to apply Corollary~\ref{coro1}.

\begin{theorem}\label{theorem:lcts}
If $\sigma \in \Sigma_{\text{leaf}}$ satisfies 
$\sigma(i,j)\ge\sigma(i,j+1)$ and $\sigma(i,j)\ge\sigma(i+1,j)$
for all $i,j\ge 1$, then $((\T_i)_{i\ge 1},P_\sigma)$ has the strong domination property.
\end{theorem}

\begin{proof}
First, we naturally extend $P_\sigma$ to a context $t\in\cT$ using equations~\eqref{lcsource1} and \eqref{lcsource2}, where we set $\sigma(0,k)=\sigma(k,0)=1$ for all $k\ge 1$
and $P_\sigma(x)=1$.
Note that $|x|=0$. Also note that $P_\sigma$ is not a probability distribution on $\cT_n$.
We have for instance $\sum_{t \in \cT_1} P_\sigma(t) = P_\sigma(f(x,a)) + P_\sigma(f(a,x)) = 2$.
We denote by $C_n$ the $n^\mathrm{th}$ Catalan number. It is well-known that $|\T_n|=C_n$ for all $n\ge 1$. We set $C_0 = 1$ and define $\lambda:\T\cup\cT\to\mathbb{R}_{>0}$ by
\[
\lambda(t)=\max\left\{\frac{1}{C_{|t|}},P_\sigma(t)\right\} .
\]
We show the four points from the strong domination property for the mapping $\lambda$.

The first point of the strong domination property, i.e., $\lambda(t)\ge P_\sigma(t)$ for all $t\in \T$, is obviously true. We now prove the second point,
i.e., $\lambda(f(s,t))\le \lambda(s)\cdot\lambda(t)$ for all $s,t\in\T$.
Assume first that $\lambda(f(s,t))=1/C_{|s|+|t|}$. The inequality $C_{m+k}\ge C_m\cdot C_k$ for all $m,k\ge 0$ yields
\begin{equation} \label{catalancase}
\frac{1}{C_{|s|+|t|}}\le \frac{1}{C_{|s|}}\cdot\frac{1}{C_{|t|}}\le\lambda(s)\cdot\lambda(t). 
\end{equation}
On the other hand, if $\lambda(f(s,t))=P_\sigma(f(s,t))$, then we have
\[
P_\sigma(f(s,t))=\sigma(|s|,|t|)\cdot P_\sigma(s)\cdot P_\sigma(t)\le P_\sigma(s)\cdot P_\sigma(t)\le\lambda(s)\cdot\lambda(t),
\]
since $0\le\sigma(i,j)\le 1$ for all $i,j$.

We now consider the third point, i.e., $\lambda(s(t))\le\lambda(s)\cdot\lambda(t)$ for all $s\in\cT$ and $t\in\T$.
Note first that the case $\lambda(s(t))=1/C_{|s|+|t|}$ follows again from equation~\eqref{catalancase}, since
$|s(t)| = |s|+|t|$. So we assume that $\lambda(s(t))=P_\sigma(s(t))$. 
Let $k$ be the length (measured in the number of edges) from the root of $s$ to 
the unique $x$-labelled node in $s$.

We show $P_\sigma(s(t)) \leq P_\sigma(s) \cdot P_\sigma(t)$ by induction over $k \geq 0$.
If $k=0$ then $s=x$ and we get
$P_\sigma(s(t)) = P_\sigma(t) = P_\sigma(x) \cdot P_\sigma(t) = P_\sigma(s) \cdot P_\sigma(t)$.
Let us now assume that $k \geq 1$. Then, $s$ must have the form $s = f(u, v)$.
Without loss of generality assume that $x$ occurs in $u$; 
the other case is of course symmetric.
Therefore, we have $s(t) = f( u(t), v)$. 
The tree $u(t)$ fulfills the induction hypothesis and therefore $P_\sigma(u(t)) \le P_\sigma(u)\cdot P_\sigma(t)$. 
Moreover, we have  $\sigma(|u(t)|, |v|) = \sigma(|u|+|t|, |v|)
\leq \sigma(|u|, |v|)$.
We get
\begin{eqnarray*}
P_\sigma(s(t)) = P_\sigma(  f( u(t), v) ) &=& P_\sigma(u(t) ) \cdot \sigma(|u(t)|, |v|) \cdot P_\sigma(v) \\
& \leq &
P_\sigma(t)\cdot P_\sigma(u) \cdot \sigma(|u|, |v|)  \cdot P_\sigma(v) \\ &=& P_\sigma(t)\cdot P_\sigma(s) .
\end{eqnarray*}
For the fourth property we show  $\sum_{t\in\T_n\cup\cT_n}\lambda(t)\leq 8n-2$
for all $n\ge 1$. First, we have 
\begin{equation}\label{sumlambdatrees}
\sum_{t\in\T_n}\lambda(t)\le \sum_{t\in\T_n}\left(C_n^{-1}+P_\sigma(t)\right)=2.
\end{equation}
We show that $\sum_{t\in\cT_n}\lambda(t)\leq 8n-4$.
Every tree $t \in \T_n$ has $2n-1$ nodes. 
Let $v$ be a node of $t$. We obtain two contexts
from $t$ and $v$ by replacing in $t$ the subtree $t[v]$ by either
$f(x,t[v])$ or $f(t[v],x)$. Let us denote the resulting contexts by $t_{v,1}$ and $t_{v,2}$.
We have $\lambda(t_{v,1}) = \lambda(t_{v,2}) = \lambda(t)$ for every node $v$ of $t$.
Moreover, for every context $t \in \cT_n$ there exists a tree $t' \in \T_n$ and a node 
$v$ of $t'$ such that $t'_{v,1} = t$ or $t'_{v,2} = t$ (depending on whether $x$ is the 
left or right child of its parent node in $t$).
Hence, we get 
\begin{equation}\label{treevscontext}
\sum_{t\in\cT_n}\lambda(t)\le (4n-2)\cdot\sum_{t\in\T_n}\lambda(t). 
\end{equation}
Together with equation~(\ref{sumlambdatrees}) we have $\sum_{t\in\T_n\cup\cT_n}\lambda(t)\le 8n-2$.
\iffalse
For each tree $t\in\T_n$ we have $2n-2$ edges in $t$.
Let $(v,v')$ be an edge of $t$, i.e., $v$ is the parent node of $v'$.
We can transform $t$ into a context $t'\in\cT_n$ with $P_\sigma(t)=P_\sigma(t')$ by adding a new node $u$, which replaces $v'$ as a child of $v$ and in return is the new parent node of $v'$.
The other child of $u$ is the parameter of $t'$.
Since we have the choice to making $v'$ the left or right child of $u$,
this yields two different contexts for each edge in $t\in\T_n$.
Additionally, the tree $t$ could also be transformed into two different contexts by adding a new root node $u$ and making the old root of $t$ either the left or right child of $u$.
Again, the other child of $u$ is the parameter.
In total, this yields for each tree $t\in\T_n$ exactly $4n-2$ contexts of the same probability.
It is easy to see that each context $t'\in\cT_n$ is covered by exactly one of the described transformations from a tree $t$. 
It follows that
\[
\sum_{t\in\cT_n}\lambda(t)\le (4n-2)\cdot\sum_{t\in\T_n} \lambda(t) \leq 8n-4.
\]
Together with equation~\eqref{sumlambdatrees} we get $\sum_{t\in\T_n\cup\cT_n}\lambda(t)\le 8n-2$.
\fi
\end{proof}

\begin{example}
An example for a leaf-centric tree source $((\T_i)_{i\ge 1},P_\sigma)$, where $\sigma \in \Sigma_{\text{leaf}}$ satisfies
$\sigma(i,j)\ge\sigma(i,j+1)$ and $\sigma(i,j)\ge\sigma(i+1,j)$
for all $i,j\ge 1$, is the famous {\em binary search tree model}; see \cite{KiefferYS09} for 
an investigation in the context of information theory. It is obtained by setting
$\sigma(i,j) = 1/(i+j)$.
\end{example}

\subsection{Depth-centric binary tree sources} \label{sec-depth-centric2}

Recall the definition of the class of mappings $\Sigma_{\text{depth}}$ by equations
\eqref{eq-sigma-depth}, \eqref{dcsource1}, and \eqref{dcsource2} in 
Section~\ref{sec-depth-centric}, and the corresponding class of depth-centric tree sources.
In this section, we state a condition on the mapping $\sigma \in \Sigma_{\text{depth}}$ that enforces the strong domination
property for the depth-centric tree source $( (\T_{i})_{i \geq 1}, P_\sigma)$. This allows again to apply Corollary~\ref{coro1}.

\begin{theorem}\label{theorem:dcts}
If $\sigma \in \Sigma_{\text{depth}}$ satisfies
$\sigma(i,j)\ge\sigma(i,j+1)$ and $\sigma(i,j)\ge\sigma(i+1,j)$
for all $i,j\ge 0$, then $((\T^i)_{i\ge0},P_\sigma)$ has the strong domination property.
\end{theorem}

\begin{proof}
Recall that the depth of a context $t\in\cT$ is defined as the depth of the tree $t(a)$.
Using this information, we extend $P_\sigma$ to a context $t\in\cT$ using equations~\eqref{dcsource1}
and \eqref{dcsource2}, where we set $P_\sigma(x) = 1$.
Similarly to the proof of Theorem~\ref{theorem:lcts} for leaf-centric tree sources, we define
\[
\lambda(t)=\max\left\{\frac{1}{C_{|t|}},P_\sigma(t)\right\} .
\]
The first point of the strong domination property, i.e., $\lambda(t)\ge P_\sigma(t)$ for all $t\in \T$,
follows directly from the definition of $\lambda$.
Now we prove the second point,
i.e., $\lambda(f(s,t))\le \lambda(s)\cdot\lambda(t)$ for all $s,t\in\T$.
The case $\lambda(f(s,t))=1/C_{|s|+|t|}$ is covered by equation~\eqref{catalancase}.
If otherwise $\lambda(f(s,t))=P_\sigma(f(s,t))$, then
\[
P_\sigma(f(s,t))=\sigma(d(s),d(t))\cdot P_\sigma(s)\cdot P_\sigma(t)\le P_\sigma(s)\cdot P_\sigma(t)\le\lambda(s)\cdot\lambda(t),
\]
since $0\le\sigma(i,j)\le 1$ for all $i,j$.

Consider now the third point, i.e., $\lambda(s(t))\le\lambda(s)\cdot\lambda(t)$ for all $s\in\cT$ and $t\in\T$.
Note that the case $\lambda(s(t))=1/C_{|s|+|t|}$ follows again from equation~\eqref{catalancase}.
Hence, we can assume that $\lambda(s(t))=P_\sigma(s(t))$. 
Again, we prove $P_\sigma(s(t)) \leq P_\sigma(s) \cdot P_\sigma(t)$ by induction over the length $k\ge 0$
(measured in the number of edges) from the root of $s$ to the unique $x$-labelled node in $s$.
If $k=0$ then $s=x$ and $P_\sigma(s)=1$, which gives us
$P_\sigma(s(t)) = P_\sigma(t) = P_\sigma(s) \cdot P_\sigma(t)$.
We now assume $k \geq 1$ and $s = f(u,v)$.
Without loss of generality assume that $x$ occurs in $u$; 
the other case is symmetric.
Therefore, we have $s(t) = f( u(t), v)$. 
We apply the induction hypothesis to the tree $u(t)$, which yields $P_\sigma(u(t)) \le P_\sigma(u)\cdot P_\sigma(t)$. 
Moreover, since $d(u)\le d(u(t))$ we have  $\sigma(d(u(t)), d(v)) \le \sigma(d(u), d(v))$.
It follows that
\begin{eqnarray*}
P_\sigma(s(t)) = P_\sigma(  f( u(t), v) ) &=& P_\sigma(u(t) ) \cdot \sigma(d(u(t)), d(v)) \cdot P_\sigma(v) \\
& \leq &
P_\sigma(t)\cdot P_\sigma(u) \cdot \sigma(d(u), d(v))  \cdot P_\sigma(v) \\ &=& P_\sigma(t)\cdot P_\sigma(s) .
\end{eqnarray*}
For the fourth property we show  $\sum_{t\in\T_n\cup\cT_n}\lambda(t)\leq 4n^2 + 3n - 1$ for all $n\ge 1$.
First, we have 
\begin{equation}\label{sumlambdatrees2}
\sum_{t\in\T_n}\lambda(t)\le \sum_{t\in\T_n}\left(C_n^{-1}+P_\sigma(t)\right)=1+\sum_{t\in\T_n}P_\sigma(t).
\end{equation}
Note that $P_\sigma$ is not a probability distribution on $\T_n$ since this section deals with depth-centric tree sources.
But for each tree $t\in \T_n$ we have $\lceil\log_2(n)\rceil\le d(t)\le n-1$, which yields
\[
\sum_{t\in\T_n}P_\sigma(t)\le \sum_{i=\lceil\log_2(n)\rceil}^{n-1}\sum_{t\in\T^i}P_\sigma(t)=n-\lceil\log_2(n)\rceil\le n.
\]
Together with equation~(\ref{sumlambdatrees2}) we get $\sum_{t\in\T_n}\lambda(t)\le n+1$.
The remaining part $\sum_{t\in\cT_n}\lambda(t)$ can be estimated with help of equation~(\ref{treevscontext}) from the corresponding part in the proof of Theorem~\ref{theorem:lcts}.
In total, we have
\[
\sum_{t\in\T_n\cup\cT_n}\lambda(t)\le (4n-1)\sum_{t\in\T_n}\lambda(t)\le (4n-1)(n+1) = 4n^2 + 3n - 1.
\]
\end{proof}

\section{Future research}

We plan to investigate, whether the strong domination property can be shown also for other classes
of tree sources. An interesting class are the tree sources derived from stochastic context-free grammars
\cite{MillerO92}. Another interesting question is, whether convergence rate of $O(\log \log i / \log i)$ in Corollary~\ref{coro1}
can be improved to $O(1 / \log i)$. In the context of grammar-based string compression, such an improvement has 
been accomplished in \cite{KiYa02}.

%\bibliographystyle{plain}
%\bibliography{bib}
\def\cprime{$'$} \def\cprime{$'$}

\end{document}